\documentclass{rmmcart}

\usepackage{amssymb}
\usepackage{nccmath}
\usepackage{mathtools}
\usepackage{mathabx}
\usepackage{hyperref}
\usepackage{comment}

\newtheorem{example}{\textbf{Example}}[section]
\newtheorem{remark}{\textit{Remark}}[section]
\newtheorem{lemma}{Lemma}[section]
\newtheorem{corollary}{Corollary}[section]
\newtheorem{proposition}{Proposition}[section]
\newtheorem{theorem}{Theorem}
\numberwithin{theorem}{section}

\title[Volterra integral equations with sum kernels]{On the solutions of linear Volterra equations of the second kind with sum kernels}
\author{Pierre-Louis Giscard}
\address{Universit\'e du Littoral C\^{o}te d'Opale, EA2597 - LMPA - Laboratoire de Math\'ematiques Pures et Appliqu\'ees Joseph Liouville, Calais, France}
\email{giscard@univ-littoral.fr}
\thanks{P.-L. G. thanks Prof. A. Slav\'ik for a fruitful discussion pertaining to product integration as well as an anonymous referee for his/her careful reading of the manuscript. 
This research is supported by the ANR JCJC \textsc{Alcohol} project, ANR-19-CE40-0006.}

\date{\today}



\begin{document}
\begin{abstract}
We consider a linear Volterra integral equation of the second kind with a sum kernel $K(t',t)=\sum_i K_i(t',t)$ and give the solution of the equation in terms of solutions of the separate equations with kernels $K_i$, provided these exist. As a corollary, we obtain a novel series representation for the solution with improved convergence properties. 
We illustrate our results with examples, including the first known Volterra equation solved by Heun's confluent functions. This solves a long-standing problem pertaining to the representation of such functions. The approach presented here has widespread applicability in physics via Volterra equations with degenerate kernels. 
\end{abstract}
\maketitle

\section{Introduction}
\subsection{Context and background}
Inhomogenous linear Volterra integral equations of the second kind  are equations in an unknown function of two  variables $f(t',t)$ given by
\begin{equation}\label{Volterra}
f(t',t) = g(t',t)+\int_t^{t'}K(t',\tau)f(\tau,t)d\tau.
\end{equation}
In this expression $g$ and $K$ are arbitrary functions of two variables, $g$ being termed the \emph{inhomogeneity} and $K$ the \textit{kernel} of the equation, respectively. Such equations are of paramount importance in physics, e.g. in designing quantum computing applications \cite{Kayanuma1994, Angelo2005, Zeuch2018, Schmidt2018}. In this case the desirability of analytical expansions for the solution is such that purely numerical approaches are of secondary relevance.\\


In this work, we focus on a particular type of kernels, which we call \textit{sum kernels}, which can be expressed as
$$
K(t',t):=\sum_{i=1}^d K_i(t',t),
$$
and such that the solutions of all the separate Volterra equations with kernels $K_i$ exist. A widespread example of sum kernels of particular interest are the \text{separable} kernels (also called \emph{degenerate} kernels), for which all the $K_i$ satisfy $K_i(t',t)=a_i(t') b_i(t)$. 
From an analytical point of view, inhomogenous linear Volterra equations with general sum kernels are solved indirectly through transformations mapping the equation into a system of coupled linear differential equations with non-constant coefficients \cite{cerha1972, Gripenberg1990, razdolsky2017, Polyanin2008}. The solutions of such systems are in fact themselves non-obvious. To make matter worse, this process is circular, pointless and unsatisfying for (quantum) physics applications for which the Volterra equations originate from reductions of an initial system of coupled linear differential equations! To the best of our knowledge,  whenever $i>1$ there is no known closed analytical form for the solution of a linear Volterra integral equation of the second kind with \emph{arbitrary}\footnote{Special solutions for certain separable kernels have of course been found, we refer the reader to \cite{Polyanin2008} Part I, Chapter 11, for a number of these.} sum kernel, even in the simpler cases of a separable kernel such that all $K_i$ are bounded over all compact subintervals of $I^2$ and smooth functions in both $t$ and $t'$. In all likeliness, such a form does not \textit{in general} exist since examples where the solution is a higher transcendental function have already been found \cite{xie2010}. 

In any case, the best possible analytical approach is therefore to provide a series representation of the solution with `good' convergence and truncation properties. Taylor polynomial expansion have been proposed \cite{sezer1994} but the Neumann series obtained via Picard iteration gives a simpler starting point in this quest, as it is amenable to exact, systematic re-summations based on set theory. It is the purpose of this note to present the resulting re-summed series. For the many numerical techniques developed so far to solve Volterra equations as well as theoretic considerations pertaining to the existence uniqueness and smoothness of the solutions, we refer to \cite{Linz1985, hackbusch1995}.

\subsection{Standard Neumann series formulation of the solutions}
It is convenient to introduce first a short hand notation for integrals such as the one appearing in Eq.~(\ref{Volterra}). 

Let $t$ and $t'$ be two real variables.
Let $f(t',t)$ and $g(t',t)$ be time-dependent generalized functions of the form 
\begin{align*}
f(t',t)=\tilde{f}(t',t)\Theta(t'-t)+\tilde{f}_0(t',t)\delta(t'-t),\\
g(t',t)=\tilde{g}(t',t)\Theta(t'-t)+\tilde{g}_0(t',t)\delta(t'-t),
\end{align*}
with $\tilde{f}$, $\tilde{f}_0$, $\tilde{g}$ and $\tilde{g}_0$ are smooth functions in both variables. Here, $\delta(t'-t)$ is the Dirac delta distribution and $\Theta(\cdot)$ designates the Heaviside theta function, with the convention $\Theta(0)=1$.
Then we define the convolution-like $\ast$-product between $f(t',t)$ and $g(t',t)$ as
\begin{equation}\label{eq:def:*}
  \big(f \ast g\big)(t',t) := \int_{-\infty}^{\infty} f(t',\tau) g(\tau, t) \, \text{d}\tau,
\end{equation}
which evaluates to \cite{schwartz1978}
\begin{align*}
 \big(f \ast g\big)(t',t)&=\tilde{f}_0(t,t)\tilde{g}_0(t,t)\delta(t'-t)+\\
 &\hspace{-15mm}\left(\int_{t}^{t'}\tilde{f}(t',\tau)g(\tau,t)d\tau+\tilde{f}(t',t)\tilde{g}_0(t',t)+\tilde{f}_0(t,t)\tilde{g}(t',t)\right)\Theta(t'-t).
\end{align*}
This definition gives the identity element with respect to the $\ast$-product as the Dirac delta distribution, thus denoted $1_\ast:=\delta(t'-t)$. As a case of special interest for the theory of Volterra equations, consider the special cases where $f(t',t):=\tilde{f}(t',t)\Theta(t'-t)$ and $g(t',t):=\tilde{g}(t',t)\Theta(t'-t)$. From now on, the tilde notation indicates that  $\tilde{f}$ is an \emph{ordinary} function. With the objects $f$ and $g$, the $\ast$-product evaluates to 
\begin{align*}
  \big(f \ast g\big)(t',t) &= \int_{-\infty}^{\infty} \tilde{f}(t',\tau) \tilde{g}(\tau, t)\Theta(t'-\tau)\Theta(\tau-t) \, \text{d}\tau,\\ &=\Theta(t'-t)\int_t^{t'} \tilde{f}(t',\tau) \tilde{g}(\tau, t) \, \text{d}\tau,
\end{align*}
which makes calculations involving such functions easier to carry out, while the form Eq.~(\ref{eq:def:*}) allows a rigorous treatment of calculations involving other distributions such as $1_\ast$. In addition, this shows that for functions of the form of $f$ and $g$ above, if $t$ and $t'$ belong to a specific interval $I\subseteq \mathbb{R}$, then the value of $(f\ast g)(t',t)$ depends only on those of $\tilde{f}$ and $\tilde{g}$ over $I^2$. This remark is important when considering functions $\tilde{f}$ and $\tilde{g}$ that are unbounded on $\mathbb{R}^2$ but might still be bounded over some a region of interest $I^2$.
To benefit from this observation, in the following we often require $\tilde{f}$ and $\tilde{g}$ to be bounded over some region of interest $I^2$ to ensure boundedness of $f\ast g$ over the same region, rather than assume boundedness over $\mathbb{R}^2$.

The $\ast$-product is also well defined for functions which depend on less than two  variables. Let $m(t')$ be a generalized function that depends on only one time variable and $g(t',t)$ be as above.  Then
\begin{align*}
\big(m \ast g\big)(t',t)&= m(t')\int_{-\infty}^{+\infty}  g(\tau, t) \, \text{d}\tau,\\
\big(g\ast m\big)(t',t)&=\int_{-\infty}^{+\infty}  g(t',\tau)m(\tau) \, \text{d}\tau.
\end{align*}
In other terms, the time variable of $m(t')$ is treated as the left variable of a generalized function depending on two variables. This observation extends straightforwardly to constant functions.\\

It is well known that linear Volterra equations of the second kind are solvable using Picard-iterations \cite{Tricomi1985, Gripenberg1990}. The underlying principle is simple: noting that, in $\ast$-product notation, the equation takes on the form
$$
f=g+K\ast f,
$$ 
for kernels of the type $K(t',t)=\tilde{K}(t',t)\Theta(t'-t)$,  we deduce that 
$$
f=g+K\ast f = g+K \ast(g+K\ast f ) = g+K\ast g + K\ast K \ast f =\cdots 
$$ 
Continuing this process yields the Neumann series representation of the solution
\begin{equation}\label{Picard}
f = \left(\sum_{n\geq 0} K^{\ast n}\right)\ast g,
\end{equation}
where $K^{\ast 0}=1_\ast$ is the Dirac delta distribution. 
The same solution holds for matrix-valued Volterra integral equations of the second kind, where the generalized functions $f$, $g$ and $K$ are time-dependent matrices with distributional entries $\mathsf{F}$, $\mathsf{G}$ and $\mathsf{K}=\tilde{\mathsf{K}}(t',t)\,\Theta(t'-t)$. In physical contexts, this was first uncovered by F. Dyson \cite{dyson1952}, who called the solution the time- or path-ordered exponential of $\mathsf{G}$. These objects are of fundamental importance in the field of quantum dynamics, the time-ordered exponential of the Hamiltonian operator dictating the evolution of a quantum system driven by time-varying forces. There, the question of solving Eq.~(\ref{Volterra}) with separable kernels is crucial \cite{Giscard2015}.\\

The solution presented in Eq.~(\ref{Picard}) can also be cast as a $\ast$-resolvent \cite{Linz1985, Tricomi1985, Gripenberg1990}, that is 
\begin{equation}\label{fgRK}
f = \left(1_\ast-K\right)^{\ast -1}\ast g,
\end{equation}
where the inverse is taken with respect to the $\ast$-product. This follows from the observation that the Neumann series Eq.~(\ref{Picard}) is convergent provided $K(t',t)=\tilde{K}(t',t)\Theta(t'-t)$ is such that $\tilde{K}(t',t)$ is an ordinary function of the two variables bounded over all compact subintervals of the interval of interest $(t',t)\in I^2\subseteq \mathbb{R}^2$. Together with the form of Eq.~(\ref{Picard}), Eq.~(\ref{fgRK}) above shows that the whole difficulty in calculating $f=(1_\ast-K)\ast g$ lies in finding the $\ast$-resolvent of $K$, denoted from now on $R_K:=\left(1_\ast-K\right)^{\ast -1}$.

A $\ast$-resolvent such as $R_K$ itself satisfies a linear integral Volterra equation of the second kind with kernel $K$ and inhomogeneity $1_\ast=\delta(t'-t)$, as implied by the usual properties of resolvents. Indeed, since $K\ast R_K= R_K-1_\ast$, we have
$$
R_K=1_\ast+ K\ast R_K.
$$ 
While the Neumann series representation of $R_K=1_\ast+K+K^{\ast2}+\cdots $ is guaranteed to exist and converges for kernels $K=\tilde{K}(t',t)\Theta(t'-t)$ with $\tilde{K}(t',t)$ bounded over all compact subintervals of $I^2$, the speed of convergence and quality of the analytical approximations obtained by truncating this series can be very poor. 
Fortunately, generic properties of $\ast$-resolvents allow exact and systematic re-summations of this series that not only speed-up its convergence but ultimately express $R_K$ in terms of the $\ast$-resolvents $R_{K_i}$ whenever $\tilde{K}(t',t)=\sum_{i}\tilde{K}_i(t',t)$. This approach is particularly very suited to separable kernels for which $\tilde{K}_i(t',t)=\tilde{a}_i(t')\tilde{b}_i(t)$ and thus every $R_{K_i}$ is known exactly.\\

%
%
\section{Algebraic properties of resolvents}
\subsection{Properties of ordinary resolvents}
Let us start by coming back to basic properties of resolvents and inverses. For example, considering an ordinary inverse and $u$ and $v$ two formal variables. We have
\begin{subequations}\label{IterationFormComm}
\begin{equation}
\frac{1}{1-u-v} = \frac{1}{1-u}\times\frac{1}{1-v}+\frac{uv}{(1-u)(1-v)}\times \frac{1}{1-u-v},\label{Form1comm} 
\end{equation}
so that an immediate iterative approach to calculating $1/(1-u-v)$ is 
\begin{align}
\frac{1}{1-u-v} &= \frac{1}{1-u}\times\frac{1}{1-v}+\frac{uv}{(1-u)(1-v)}\times\\
&\hspace{15mm}\left(\frac{1}{1-u}\times\frac{1}{1-v}
+\frac{uv}{(1-u)(1-v)}\cdots\right).\nonumber
\end{align}
\end{subequations}
This iteration leads to the formal series,
\begin{align*}
\frac{1}{1-u-v} =\sum_{k=0}^\infty\frac{(uv)^k}{(1-u)^k(1-v)^k}\frac{1}{1-u}\frac{1}{1-v}.
\end{align*}
These observations extend to resolvents with respect to non-commutative products, in particular the $\ast$-product as we show explicitly below.

\subsection{Consequences for $\ast$-resolvents}\label{Sumof2}
Inspired by the procedure presented above, we may now express the $\ast$-resolvent $R_{K}$ for a kernel $K=\sum_{i=1}^d K_i$ which is a sum of  $d$ kernels $K_i$ in terms of $R_{K_i}$, assuming only that the individual $R_{K_i}$ exist and that the overall kernel $K$ is bounded  over all compact subintervals of $I^2$. We begin by an explicit statement for $\ast$-resolvents of the results of the previous section, but this time involving the sum of any number $d\geq 2$ of kernels.
\begin{theorem}[Sum of $\ast$-resolvents]\label{Kleene}
Let $(t',t)\in I^2$ be two variables and $K_{i}(t',t)=\tilde{K}_i(t',t)\Theta(t'-t)$, $i=1,\cdots,d$ be $d\geq 2$ kernels such all $\tilde{K}_i(t',t)$ are bounded over all compact subintervals of $I^2$.  Let $K$ be the corresponding sum kernel $K:=\sum_{i=1}^dK_i$.
Then all $\ast$-resolvents $R_{K_i}$ exist and the $\ast$-resolvent $R_K$ of $K$ satisfies
\begin{align}\label{IterationForm}
R_K=T\ast R_K+\bigast_{i=1}^d R_{K_{i}}, 
\end{align}
where  
$$
T:=1_\ast-\bigast_{i=1}^d R_{K_{i}}\ast(1_\ast - K),
$$
and there exist and ordinary function $\tilde{T}(t',t)$ bounded over all compact subintervals of $I^2$ such that $T(t',t)=\tilde{T}(t',t)\Theta(t'-t)$
This shows that $R_K$ satisfies an inhomogeneous linear Volterra integral equation of the second kind with kernel  $T$ and inhomogeneity $\bigast_{i=1}^d R_{K_{i}}$. Consequently, $R_K$ is given by the convergent Neumann series
\begin{equation}\label{RkSeries}
R_K=\sum_{k=0}^\infty T^{\ast k}\ast\bigast_{i=1}^d R_{K_{i}},
\end{equation}
where $T^{\ast 0}=1_\ast$.
\end{theorem}
%

\begin{proof}
The result follows from direct algebraic manipulations. With $T$ as defined above we have
$$
T\ast R_K = 1_\ast \ast R_K-\bigast_{i=1}^d R_{K_i}\ast (1_\ast - K)\ast R_K
$$
but $(1_\ast - K)\ast R_K=1_\ast$, $1_\ast \ast R_K=R_K$ and $\bigast_{i=1}^d R_{K_i}\ast 1_\ast = \bigast_{i=1}^d R_{K_i}$. Thus
$$
T\ast R_K = R_K-\bigast_{i=1}^d R_{K_i}
$$
so that $T\ast R_K+\bigast_{i=1}^d R_{K_i}=R_K$. To obtain the series representation of $R_K$ we perform a Picard iteration, assuming convergence for the time being. Iteratively replacing $R_K$ by its expression in Eq.(\ref{IterationForm}),
\begin{align*}
 R_K&=T\ast \left(T\ast \left(T\ast (\cdots)+\bigast_{i=1}^d R_{K_i}\right)+\bigast_{i=1}^d R_{K_i}\right)+\bigast_{i=1}^d R_{K_i},\\
 &=\left(1_\ast+T+T^{\ast 2} + T^{\ast 3} + \cdots\right)\bigast_{i=1}^d R_{K_i}.
\end{align*}
To prove convergence of the above Neumann series, we first reformulate $T$ so as to show that it is of the form $\tilde{T}(t',t)\Theta(t'-t)$ by explicitly removing the distribution $1_\ast$:
\begin{align*}
T&=1_\ast - \bigast_{i=1}^d R_{K_i} +\bigast_{i=1}^d R_{K_i}\ast K,\label{TForm}\\
&=1_\ast - \bigast_{i=1}^d (1_\ast+K_i\ast R_{K_i})+\bigast_{i=1}^d R_{K_i}\ast K.
\end{align*}
Expanding the above products demonstrates that no isolated $1_\ast$ remains.
Now because all $\tilde{K}_i$ are bounded over  all compact subintervals of $I^2$, all the Neumann series $\sum_k K_i^{\ast k}$ converge and thus $R_{K_i}=1_\ast + \tilde{R}_{K_i}(t',t)\Theta(t'-t)$ with $\tilde{R}_{K_i}$ bounded  all compact subintervals of $I^2$. Since furthermore $K=\tilde{K}(t',t)\Theta(t'-t)$ is also bounded, the above result implies that $T=\tilde{T}(t',t)\Theta(t'-t)$ with $\tilde{T}$ bounded over  all compact subintervals of $I^2$. We now observe that since $\Theta(t'-t)^{\ast k } =(t'-t)^{k-1}/(k-1)!\times \Theta(t'-t)$, $T^{\ast k}$ is bounded over  all compact subintervals of $I^2$ by $C_T^k(t'-t)^{k-1}/(k-1)!$, where $$
C_T=\sup_{t',t\in I}|\tilde{T}(t',t)|,
$$ 
is a constant that bounds $T$ over $I^2$. Thus the Neumann series $\sum_k T^{\ast k}$ is convergent. 
\end{proof}

A few remarks are now in order regarding the results of Theorem~\ref{Kleene}:
\begin{remark}\label{Order}
The order in which the individual resolvents $R_{K_i}$ appear in the above products is irrelevant so long as the quantity $T$ is defined correspondingly, i.e. with the same order of the products. Consider an example involving 2 resolvents $K_1$ and $K_2$. Let $R_K=(1_\ast - K_1-K_2)^{\ast-1}$ and
\begin{align*}
T_{12}&=1_\ast-R_{K_1}\ast R_{K_2}\ast(1_\ast - K_1-K_2),\\
T_{21}&=1_\ast-R_{K_2}\ast R_{K_1}\ast(1_\ast - K_1-K_2).
\end{align*}
Then we verify by direct substitution that 
$$
R_K=T_{12}\ast R_K+K_1\ast K_2=T_{21}\ast R_K+K_2\ast K_1.
$$
In general, equally valid formulas are obtained upon permuting the individual resolvents  in the theorem results,  owing ultimately to the exchange symmetry of the resolvent $R_{K}$.
\end{remark}

\begin{remark}\label{FormT}
The quantity $T$ arising in the expansion of the $\ast$-resolvent $R_K$ of a sum kernel $K=\sum_{i=1}^d K_i$ admits many alternative forms that can prove useful in determining bounds on $\sup_{t',t\in I^2}|T(t',t)|$. To see this, we first consider an example, that of the sum of three kernels $K=K_1+K_2+K_3$. Then 
\begin{align*}
T&=1_\ast - R_{K_3}\ast R_{K_2}\ast R_{K_1} \ast (1_\ast - K_1 - K_2 -K_3),\\
&= R_{K_3}\ast R_{K_2}\ast R_{K_1}\ast\\
&\hspace{10mm}\big((1_\ast-K_1)\ast (1_\ast - K_2)\ast(1_\ast - K_3)-1_\ast + K_1+K_2+K_3\big),\\
&=R_{K_3}\ast R_{K_2}\ast R_{K_1}\ast\big(K_1\ast K_2+K_1\ast K_3+K_2\ast K_3 - K_1\ast K_2\ast K_3\big).
\end{align*}
Here, the crux of the proof is to write $1_\ast=R_{K_3}\ast R_{K_2}\ast R_{K_1}\ast(1_\ast-K_1)\ast (1_\ast - K_2)\ast(1_\ast - K_3)$. Of course, we can always do so, which yields the following general expression for the quantity $T$ obtained when expressing the $\ast$-resolvent of $K=\sum_{i=1}^d K_i$,
$$
T=\bigast_{i=1}^d R_{K_i} \ast \sum_{n=2}^d (-1)^n \sum_{i_1<i_2<\cdots<i_n}K_{i_1}\ast K_{i_2}\ast \cdots \ast K_{i_n}
$$
with the first product understood as $\bigast_{i=1}^d R_{K_i}=R_{K_d}\ast \cdots R_{K_2}\ast R_{K_1}$. An equally valid expression holds upon reversing the orders in the products above, in accordance with Remark~\ref{Order}.

\end{remark}

The advantages of the theorem's results for calculating $\ast$-resolvents are as follows:\\[-1.5em]
\begin{enumerate} 
\item[i)] They hold for any number of kernels functions in the overall resolvent $R_{K}=(1_\ast-K_1-K_2-K_3-\cdots)^{\ast-1}$ so long as each is of the form $K_i(t',t)=\tilde{K}_i(t',t)\Theta(t'-t)$ with $\tilde{K}_i$ an ordinary function bounded over  all compact subintervals of$I^2$;
\item[ii)] Solving a Volterra equation with sum kernel only necessitates knowing the $\ast$-resolvents of the individual kernels $K_i$, a huge advantage in the case of degenerate kernels;
\item[iii)] The theorem provides a fully explicit series which is a re-summed form of the original Neumann series representation of the solution; 
\item[iv)] The theorem yields an improvement of the convergence speed over the original Neumann series especially when one or more kernel $K_i$ dominates over the others, see \S\ref{Converge} and the examples of \S\ref{section:Examples}.
\item[v)] It remains valid should all ordinary functions $\tilde{K}_1,\,\tilde{K}_2,\cdots$ be time-dependent matrices\footnote{In this case, $\ast$-resolvents such as $(1_\ast-\mathsf{V})^{\ast-1}(t',t)$ should be understood as the time derivative with respect to $t'$ of the ordered exponential of the time-dependent matrix $\mathsf{V}(t',t)=\tilde{\mathsf{V}}(t',t)\Theta(t'-t)$};\\[-.5em]
\end{enumerate}




\section{Main result: application to Volterra equations}
We may now formulate our general results for the specific purpose of solving linear inhomogenous Volterra equations of the second kind for general sum kernels. 

\begin{corollary}\label{ReSummedNeumann}
Let $I\subseteq \mathbb{R}$, $(t',t)\in I^2$, $d\in\mathbb{N}=\{1,2,3\cdots\}$ and $K(t',t) := \sum_{i=1}^{d} \tilde{K}_i(t',t)\Theta(t'-t)$ be a sum kernel, in particular all $\ast$-resolvents $R_{K_i}$ must exist. Let $g(t',t)$ be a generalized function that is not identically null over $I^2$ and let
$$
T:=1_\ast-\bigast_{i=1}^d R_{K_{i}}\ast(1_\ast - K).
$$
Then the generalized function $f$, which is solution of the linear Volterra integral equation of the second kind with kernel $K$ and inhomogeneity $g$, $f = g+K\ast f$, is given by
\begin{align*}
f &
 = \left(\sum_{k=0}^\infty T^{\ast k}\right) \bigast_{i=1}^{d}R_{K_i}\ast g,
\end{align*}
with $T^{\ast 0}=1_\ast$.
%
%
Let $f^{(n-1)}$ be the series as above truncated at order $n-1\geq 0$, that is 
$$
f^{(n-1)}:=\left(\sum_{k=0}^{n-1} T^{\ast k}\right) \bigast_{i=1}^{d}R_{K_i}\ast g
$$
Then, the truncation error is exactly
\begin{equation}\label{ErrorT}
f-f^{(n-1)} = T^{\ast n}\ast f.
\end{equation}
\end{corollary}
\begin{proof}
For the first statement, observe that $R_K$ solves the Volterra equation $R_K=1_\ast+K\ast R_K$. Thus, letting $f:=R_K\ast g$, we have  $f= (1_\ast+K\ast R_K)\ast g = g+ K \ast R_K \ast g=g+ K\ast f$, that is $f$ solves the the linear Volterra integral equation of the second kind with kernel $K$ and inhomogeneity $g$. This observation, together with the form of $R_K$ as given by Eq.~(\ref{RkSeries}) of Theorem~\ref{Kleene}, yields the first result. 

For the truncation error, we obtain it by direct calculation
\begin{align*}
T^{\ast n} \ast f&=T^{\ast n}\ast \sum_{k=0}^{\infty} T^{\ast k} \bigast_{i=1}^{d}R_{K_i}\ast g,\\
&=\sum_{k=n}^{\infty} T^{\ast k} \bigast_{i=1}^{d}R_{K_i}\ast g,\\
&=\left(\sum_{k=0}^{\infty} T^{\ast k}-\sum_{k=0}^{n-1} T^{\ast k}\right) \bigast_{i=1}^{d}R_{K_i}\ast g,\\
&=f-f^{(n-1)}.
\end{align*}
%
%
\end{proof}

Corollary~\ref{ReSummedNeumann} is especially well suited to separable kernels, for in such cases all $R_{K_i}$ are exactly available. Recall that a separable kernel takes on the form 
\begin{align}\label{SepKernel}
&K(t',t)=\tilde{K}(t',t)\Theta(t'-t),\qquad\tilde{K}(t',t) = \sum_{i=1}^d \tilde{K}_i(t',t),\\
&\text{with }~\tilde{K}_i(t',t):= \tilde{a}_i(t') \tilde{b}_i(t).\nonumber
\end{align}
Here all $\tilde{a}_i(t')$ and $\tilde{b}_i(t)$ ordinary functions of a single time variable.
For the sake of simplicity, we assume  that all $\tilde{a}_i$ and $\tilde{b}_i$ are bounded and smooth over $I$. Although these conditions can be relaxed following \cite{Linz1985}, they are satisfied in e.g. all quantum physics applications. Now, all $\ast$-resolvents $R_{K_i}:=\big(1_\ast-K_i\big)^{\ast-1}$ are available in closed form. This result is almost certainly already known although we could not locate it in the literature. We provide it here with a proof for the sake of completeness:

%

\begin{proposition}\label{SingleSeparable}
Let $I\subset \mathbb{R}$ and let $(t',t)\in I^2$ be two  variables and let $g(t',t)$ be a generalized function of $t',t$. Let $f(t',t)=\tilde{f}(t',t)\Theta(t'-t)$ be a function of $t',t$ over $I^2$ and $K(t',t):=\tilde{a}(t')\tilde{b}(t)\Theta(t'-t)$. Let $\tilde{\alpha}:=\int \tilde{K}(\tau,\tau) d\tau = \int \tilde{a}(\tau)\tilde{b}(\tau) d\tau$. Then the solution $f$ of the linear Volterra equation of the second kind $f = g + K\ast f$ with kernel $K$ is
\begin{align}
f(t',t)& = g(t',t) +\nonumber\\
&\hspace{5mm} \tilde{a}(t')\int_{-\infty}^{\infty}\tilde{b}(\tau) \exp\left(\int_{\tau}^{t'} \tilde{a}(\tau')\tilde{b}(\tau') d\tau'\right) \Theta(t'-\tau) g(\tau,t) \,d\tau.\label{RSingleSeparable}\end{align}
\end{proposition}
\begin{remark}\label{RemSingleSep}
In the (typical) case where $g$ itself takes on the form $g(t',t)=\tilde{g}(t',t)\Theta(t'-t)$, in the expression of Eq.~(\ref{RSingleSeparable}), $g(\tau,t)$ can be replaced with $\tilde{g}(\tau,t)$ with the outer integral running from $t$ to $t'$.
If instead one chooses $g(t',t)=\delta(t'-t)$, then the Volterra equation satisfied by $f$ reads $f=1_\ast+K\ast f$, that is $f$ is the $\ast$-resolvent of $K$, $f=R_K$ and Eq.~(\ref{RSingleSeparable}) simplifies to
\begin{align*}
  R_{K}(t',t)&=\delta(t'-t)+\tilde{a}(t')\tilde{b}(t) e^{\tilde{\alpha}(t')-\tilde{\alpha}(t)}\Theta(t'-t)\\
  &=\delta(t'-t)+\tilde{K}(t',t)e^{\tilde{\alpha}(t')-\tilde{\alpha}(t)}\Theta(t'-t).
\end{align*}
In other terms, the $\ast$-resolvent of a kernel of the form $K(t',t)=\tilde{a}(t')\tilde{b}(t)\Theta(t'-t)$ is exactly available in closed form. 
\end{remark}
%
 
\begin{proof}
We proceed by induction on the Neumann series $f=\big(\sum_n K^{\ast n}\big)\ast g$. Convergence of this series is guaranteed whenever $\tilde{a}$ and $\tilde{b}$ are bounded over all compact subintervals of $I$, however existence of the final form for $f$ is clearly independent from this assumption. In this situation this form can be understood as the analytic continuation of the original Neumann series. 

The proposition to be shown is 
\begin{align*}
P_n:=\Big\{\forall n\in\mathbb{N}:&~K^{\ast n}(t',t)=\tilde{K}^{\ast n}(t',t)\Theta(t'-t)~\text{with}\\
&\tilde{K}^{\ast n}(t',t)=\tilde{a}(t')\tilde{b}(t)\big(\tilde{\alpha}(t')-\tilde{\alpha}(t)\big)^{n-1}\!\!\times 1/(n-1)!\Big\}.
\end{align*}
We have $K^{\ast 1}=K$, which proves $P_1$. Now supposing $P_n$ true,
\begin{align*}
K^{\ast (n+1)}&=\int_{t}^{t'}\tilde{K}(t',\tau)\tilde{K}^{\ast n}(\tau,t)d\tau\,\Theta(t'-t),\\ 
&= \frac{1}{(n-1)!}\tilde{a}(t')\int_{t}^{t'}\tilde{a}(\tau)\tilde{b}(\tau)\big(\tilde{\alpha}(\tau)-\tilde{\alpha}(t)\big)^{n-1}d\tau\,\tilde{b}(t)\Theta(t'-t),
\end{align*}
where the $\tilde{K}^{\ast n}(\tau,t)$ term contributes both $\tilde{a}(\tau)$ and $\tilde{b}(t)$, while $\tilde{K}(t',\tau)$ gives rise to $\tilde{a}(t')$ and $\tilde{b}(\tau)$. Since $\tilde{a}(\tau)\tilde{b}(\tau)=\tilde{\alpha}'(\tau)$, we have that $K^{\ast (n+1)}(t',t)=\tilde{K}^{\ast (n+1)}(t',t)\Theta(t'-t)$ with
\begin{align*}
\tilde{K}^{\ast (n+1)}(t',t)&=\frac{1}{(n-1)!}\tilde{a}(t')\tilde{b}(t)\int_{t}^{t'}\tilde{\alpha}'(\tau)\big(\tilde{\alpha}(\tau)-\tilde{\alpha}(t)\big)^{n-1}d\tau,\\
&=\frac{1}{n!}\tilde{a}(t')\tilde{b}(t)\big(\tilde{\alpha}(t')-\tilde{\alpha}(t)\big)^{n}.
\end{align*}
This establishes the implication $P_n\Rightarrow P_{n+1}$ and since $P_1$ holds, $P_n$ is true for all $n\in\mathbb{N}$. Then $\sum_{n\geq 0} K^{\ast n}=1_\ast + K e^{\alpha}$ with $K e^{\alpha}=\tilde{K}e^{\tilde{\alpha}}\Theta$  and $f$ is obtained upon $\ast$-multiplying by $g$ from the right. 
\end{proof}

A considerable number of Volterra equations of important interest fall in the case where the kernel $K$ is separable: this includes the kernels encountered in the celebrated Bloch-Siegert and  Autler-Townes effects, or as shown below, for special functions such as the Heun's functions. More generally, the method of path-sums together with the results of Pleshchinskii \cite{Pleshchinskii1995}, show that the dynamics of all quantum systems driven by time-dependent forces can be determined by solely solving linear Volterra integral equation of the second kind with separable kernels. 


%

Following Aristotle who said that ``\textit{for the things we have to learn before we can do, we learn by doing}", it is essential to present examples of applications of the above approach. These are presented in Section~\ref{section:Examples} in ascending order of difficulty. 
In the remainder of the present section, we give the convergence analysis of the re-summed series.

\subsection{Convergence analysis}\label{Converge}
Here we suppose all ordinary functions $\tilde{K}_i(t',t)$ appearing in the kernel $K(t',t):=\sum_{i=1}^d \tilde{K}_i(t',t)\Theta(t'-t)=\tilde{K}(t',t)\Theta(t'-t)$ are bounded over  all compact subintervals of $I^2$ so as to guarantee convergence of the original Neumann series. In addition here we let $I=[\mathsf{T},\mathsf{T}']$ to avoid confusion with the two variables $t,t'\in I$, and thus $|I|=\mathsf{T}'-\mathsf{T}$.\\[-.5em] 

Let $f_{\text{N}}^{(n-1)}:=\sum_{k=0}^{n-1}K^{\ast k}\ast g$ be the approximation obtained from the truncating the original Neumann series at order $n-1$, $n\geq 1$. Observe that since $K(t',t)=\tilde{K}(t',t)\Theta(t'-t)$ with $\tilde{K}(t',t)$ an ordinary function, the Neumann series representation of $f$ shows that there exist ordinary functions $\tilde{f}^{(n-1)}$, $n\geq 1$ and $\tilde{f}$ such that $f^{(n-1)}(t',t)=\tilde{f}^{(n-1)}(t',t)\Theta(t'-t)$ and $f(t',t)=\tilde{f}(t',t)\Theta(t'-t)$.
Now we remark that
$$
f-f_{\text{N}}^{(n-1)}=\sum_{k=n}^\infty K^{\ast k} \ast g = K^{\ast n}\ast\sum_{k=0}^\infty K^{\ast k} \ast g = K^{\ast n} \ast f,
$$ 
defining $C_K:=\sup_{t',t\in I^2}|\tilde{K}(t',t)|$ and $C_f:=\sup_{t,t'\in I} |\tilde{f}(t',t)|$, we have the immediate bound, for $t'\geq t$, \cite{Linz1985}
$$
\sup_{t,t'\in I^2}|\tilde{f}-\tilde{f}_{\text{N}}^{(n-1)}|\leq \frac{C_f\big(C_K\,|I|\big)^{n}}{n!},
$$
which comes from the observation that the error comprises $n$, $\ast$-products.
This bound is saturated by the constant kernel $K(t',t)=C_K\Theta(t'-t)$. For the re-summed series presented here, the error produced by truncating the re-summed series at order $n-1$, $n\geq1$, is given by Eq.~(\ref{ErrorT})  
$$
f-f^{(n-1)}=T^{\ast n}\ast f.
$$
Since $T(t',t)=\tilde{T}(t',t)\Theta(t'-t)$ is bounded over all compact subintervals of $I^2$, as established by Theorem~\ref{Kleene}, there exists a constant $C_T:=\sup_{t',t\in I^2}|\tilde{T}(t',t)|$. Then we have
$$
\sup_{t,t'\in I^2}|\tilde{f}-\tilde{f}_{}^{(n-1)}|\leq \frac{C_{f}\big(C_T |I|\big)^{n}}{n!}.
$$
This establishes that the series of Corollary~\ref{ReSummedNeumann} converges faster than the original Neumann series provided $C_T<C_K$.  
Thus, to quantify any possible improvement we need to explicitly bound $C_T$. 

Given the construction of $R_K$ in terms of the $R_{K_i}$, we expect the approach proposed here to work best if one of the kernels $K_i$ is largely dominant over the others. To make this observation precise, we introduce the concept of $\ast$-domination:
\begin{definition}
Let $f(t',t)$ and $g(t',t)$ be two generalized functions of two variables defined over $I^2$. If there exist a generalized function $\epsilon(t',t)$ defined over $I^2$ such that 
$$
\epsilon\ast f= g,~\text{ with }~ C_\epsilon:=\sup_{t',t}|\epsilon(t',t)|\leq 1,
$$
then we say that $f$ is $\ast$-dominant over $g$ on $I^2$. 
The function $\epsilon$ is called the $\ast$-domination factor of $f$ over $g$. 
Introducing $C_f:=\sup_{t',t}|f(t',t)|$ and $C_g:=\sup_{t',t}|g(t',t)|$ the $\ast$-domination implies
$$
C_g \leq C_\epsilon\times C_f \times(t'-t).
$$
\end{definition}

Now suppose without loss of generality that $K_1$ is dominant over all other kernels $K_{1<i\leq d}$. Let $\epsilon_i$ be the $\ast$-domination factor of $K_1$ over $K_i$ 
Since $K=K_1+\sum_{i=2}^d K_i$, the $\ast$-domination by $K_1$ gives $K_i=\epsilon_i\ast K_1$ and therefore 
$
K=K_1+ \left(\sum_{i=2}^d \epsilon_i\right)\ast K_1.
$
In order to alleviate the notation, we define $\epsilon_K:=\sum_{i=2}^d \epsilon_i$ so that $K=K_1+\epsilon_K\ast K_1$.  
The quantity $T$ of Corollary~\ref{ReSummedNeumann} now reads
\begin{align}\label{TFormDom}
T&=1_\ast - \bigast_{i=2}^d R_{K_i}\ast R_{K_1}\ast \big(1_\ast-K_1-\epsilon_K\ast K_1\big),\nonumber\\
&=1_\ast - \bigast_{i=2}^d R_{K_i}\ast (1_\ast - R_{K_1}\ast \epsilon_K\ast K_1),\nonumber\\
&=1_\ast - \bigast_{i=2}^d R_{K_i} - \bigast_{i=1}^d R_{K_i}\ast \epsilon_K\ast K_1
\end{align}
since  $R_{K_1}\ast \big(1_\ast-K_1)=1_\ast$. These results produce a tight bound over $C_T$. To see this, we need to control the $\ast$-resolvents $R_{K_i}$:

\begin{lemma}\label{Rdomination}
 Let $K(t',t):=\tilde{K}(t',t)\Theta(t'-t)$ be a generalized function of two variables defined over $I^2$ such that $\tilde{K}$ is an ordinary bounded function over all compact subintervals of  $I^2$. Let $C_K:=\sup_{t',t\in I}|K(t',t)|$. Then the $\ast$-resolvent $R_K$ of $K$ exists and is of the form  
 $$
 R_K=1_\ast + \tilde{R}_K(t',t)\Theta(t'-t),
 $$
 with $\tilde{R}_K(t',t)$ an ordinary function bounded over all compact subintervals of $I^2$. In particular, 
 $$
\sup_{t',t\in I}|\tilde{R}_K(t',t)|\leq C_K e^{C_K |I|}.
 $$
\end{lemma}

\begin{proof}
 We have $R_K:=\sum_{n\geq 0} K^{\ast n}=1_\ast + \sum_{n\geq 1} K^{\ast n}$ converges since $K$ is bounded, which implies $\sup_{t',t\in I}| K^{\ast n}| \leq \sup_{t',t\in I} |C_K^{\ast n}|$ with
 $$
\big(C_K^{\ast n}\big)(t',t)=C_K^n \frac{(t'-t)^{n-1}}{(n-1)!},
 $$ 
 and thus
 $$
 \sup_{t',t\in I}|\tilde{R}_K|=\sup_{t',t\in I}|R_K-1_\ast|\leq \sup_{t',t\in I}\sum_{n=1}^\infty C_K^n \frac{(t'-t)^{n-1}}{(n-1)!}\leq C_K e^{C_K |I|}.
 $$
 This result is also directly obtained upon $\ast$-dominating $K$ by $C_K \Theta(t'-t)$ (this is because $C_K$ dominates $\tilde{K}(t',t)$ by definition, and $K=\tilde{K}$ whenever $t'\geq t$ while $K=0$ otherwise) and using Proposition~\ref{SingleSeparable} to get the $\ast$-resolvent of $C_K$. 
\end{proof}

Lemma~\ref{Rdomination} indicates that $R_{K_i}=1_\ast + \tilde{R}_{K_i}(t',t) \Theta(t'-t)$ and from there it follows
\begin{align}\label{Rproduct}
&\bigast_{i=1}^d R_{K_i} = 1_\ast +\sum_{n=1}^d\sum_{i_1<i_2<\cdots <i_n}\bigg(\int_{t}^{t'} \int_{\tau_1}^{t'}\cdots\int_{\tau_{n-2}}^{t'} \tilde{R}_{K_{i_1}}(t',\tau_{n-1})\cdots\\
&\hspace{15mm}\cdots \tilde{R}_{K_{i_{n-1}}}(\tau_2,\tau_1)\tilde{R}_{K_{i_n}}(\tau_1,t)d\tau_{n-1}\cdots d\tau_{2}d\tau_1\bigg)\Theta(t'-t),\nonumber
\end{align}
and here all $i_{1\leq j\leq n}$ indices are integers with $1\leq i_j\leq d$. Then Lemma~\ref{Rdomination} gives
\begin{align*}
\sup_{t',t\in I}\left|\bigast_{i=1}^d R_{K_i}-1_\ast\right|&\leq \sum_{n=1}^d\sum_{i_1<\cdots <i_n}\prod_{j=1}^n(C_{K_{i_j}})\times e^{\sum_{j=1}^n C_{K_{i_j}} |I|}.
\end{align*}
where we recall that $C_{K_i}:=\sup_{t',t\in I}|K_i|$.
The $\ast$-domination of $K_i$ by $K_1$ gives the bounds $C_{K_i}\leq C_{\epsilon_i} C_K$ and thus, with $C_\epsilon:=\max_{1<i\leq d} C_{\epsilon_i}$,
\begin{align*}
\sup_{t',t\in I}\left|\bigast_{i=1}^d R_{K_i}-1_\ast\right|&\leq\sum_{n=1}^{d} \binom{d}{n}(C_{\epsilon} C_{K_1})^n e^{n C_{\epsilon} C_K |I|},\\
&\leq\left(C_\epsilon C_{K_1}  e^{C_\epsilon C_{K_1} |I|}+1\right)^d-1
\end{align*}
and similarly 
$$
\sup_{t',t\in I}\left|\bigast_{i=2}^d R_{K_i}-1_\ast\right|\leq \left(C_\epsilon C_{K_1}  e^{C_\epsilon C_{K_1} |I|}+1\right)^{d-1}-1.
$$
Returning to Eq.~(\ref{TFormDom}), the above bounds give
\begin{align*}
C_T&\leq \left(C_\epsilon C_{K_1}  e^{C_\epsilon C_{K_1} |I|}+1\right)^{d-1}-1\\
&\hspace{15mm}+d\left(\left(C_\epsilon C_{K_1}  e^{C_\epsilon C_{K_1} |I|}+1\right)^{d}-1\right)C_\epsilon C_{K_1}.
\end{align*}
Here we used that $\sup_{t',t\in I}|\epsilon_K|\leq d\, C_\epsilon$. In particular $C_T\to 0$ as $C_\epsilon\to 0$, that is as $K\to K_1$ as expected. Whether or not this is an improvement, i.e. whether or not $C_T<C_K$, is now seen to depend on the precise relation between the $K_{i}$ and $K$ as it influence both $C_\epsilon$ and the relation between $C_K$ and $C_{K_1}$. If the kernels $K_i$ are large and of similar magnitude  (yielding $C_\epsilon\sim 1$) but such that important cancellations produce a very small overall sum kernel $K$ (for which $C_K$ is small), then clearly these cancellations are missed by Corollary~\ref{ReSummedNeumann} and convergence is slower than in the original Neumann series. At the opposite, if $K$ is largely dominated by a single kernel $K_1$, then the present approach becomes arbitrarily good as $K\to K_1$. This is for example the case in most physics applications, where $K$ results from a \textit{perturbation} of some sort around a dominant term $K_1$.

\section{Illustrative examples}\label{section:Examples}
We provide two examples of application of our main result, Corollary~\ref{ReSummedNeumann}, in ascending order of difficulty: i) solution to the linear Volterra equation of the second kind with constant kernel and ii) a new representation of Heun's confluent functions. 

\begin{example}[Constant kernel]
Let $a$ and $b$ be two constants and $K(t',t):=(a+b)\Theta(t'-t)$ be the kernel of the homogenous linear Volterra equation of the second kind $f=1_\ast + K\ast f$, that is 
$$
f(t',t)=\delta(t'-t)+\int_{t}^{t'} (a+b) \tilde{f}(\tau,t)d\tau\,\Theta(t'-t).
$$ 
The exact solution is given by the generalized function 
$$
f(t',t)=\delta(t'-t) + (a+b) e^{(a+b)(t'-t)}\Theta(t'-t).
$$
The ordinary Neumann series for this solution is the Taylor expansion
$$
f(t',t) = \delta(t'-t) + \left(a+b+ (a+b)(t'-t)+(a+b)\frac{(t'-t)^2}{2}+\cdots\right)\Theta(t'-t)
$$
Now let $K:=K_1+K_2$, $K_1(t',t):=a\Theta(t'-t)$ and $K_2(t',t):=a\Theta(t'-t)$. Using Remark~\ref{RemSingleSep} for the forms of $R_{K_1}$ and $R_{K_2}$, Corollary~\ref{ReSummedNeumann} dictates that 
\begin{align}\label{FExpressEx1}
f(t',t)&=\sum_{k=0}^\infty T^{\ast k} \ast\big(1_\ast + a e^{a(t'-t)}\Theta(t'-t)\big)\ast \big(1_\ast + b e^{b(t'-t)}\Theta(t'-t)\big).
\end{align}
Now
\begin{align*}
&\big(1_\ast + a e^{a(t'-t)}\Theta(t'-t)\big)\ast \big(1_\ast + b e^{b(t'-t)}\Theta(t'-t)\big)\\
&\hspace{5mm}=1_\ast + \left(a e^{a(t'-t)}+b e^{b(t'-t)}+\int_{t}^{t'} a e^{a(t'-\tau)}be^{b(\tau-t)} d\tau\right)\Theta(t'-t),\\
&\hspace{5mm}=1_\ast + \left(a e^{a(t'-t)}+b e^{b(t'-t)}+\frac{a b}{a-b} \left(e^{a (t'-t)}-e^{b (t'-t)}\right)\right)\Theta(t'-t),\\
&\hspace{5mm}=\Big(1_\ast + \frac{a^2 e^{a (t'-t)}-b^2 e^{b (t'-t)}}{a-b}\Theta(t'-t)\Big).
\end{align*}
Furthermore, in Eq.~(\ref{FExpressEx1}), $T$ is 
\begin{align*}
T&=1_\ast -\big(1_\ast + a e^{a(t'-t)}\Theta(t'-t)\big)\ast \big(1_\ast + b e^{b(t'-t)}\Theta(t'-t)\big)
\\
&\hspace{40mm}\ast\big(1_\ast - a\Theta(t'-t)-b\Theta(t'-t)\big)\\
&=1_\ast -\Big(1_\ast + \frac{a^2 e^{a (t'-t)}-b^2 e^{b (t'-t)}}{a-b}\Theta(t'-t)\Big)\\
&\hspace{40mm}\ast\big(1_\ast - a\Theta(t'-t)-b\Theta(t'-t)\big)\\
&=1_\ast - 1_\ast +(a+b)\Theta(t'-t)-\frac{a^2 e^{a (t'-t)}-b^2 e^{b (t'-t)}}{a-b}\Theta(t'-t)\\
&\hspace{20mm}+\int_{t}^{t'}\frac{a^2 e^{a (t'-\tau)}-b^2 e^{b (t'-\tau)}}{a-b}(a+b)d\tau
\,\Theta(t'-t)\\
&=(a+b)\Theta(t'-t)-\frac{a^2 e^{a (t'-t)}-b^2 e^{b (t'-t)}}{a-b}\Theta(t'-t)\\
&\hspace{20mm}+\frac{a+b}{a-b} \left(a e^{a
   (t'-t)}-be^{b
   (t'-t)}-(a-b)\right)\Theta(t'-t)\\
&=\frac{a b  \left(e^{a (t'-t)}-e^{b (t'-t)}\right)}{a-b}\Theta(t'-t).
\end{align*}
Therefore even at order 0, i.e. with the truncation $\sum_{k=0}^\infty T^{\ast k} \rightarrow T^{\ast 0}=1_\ast$, we get
\begin{align*}
f^{(0)}(t',t) &:= \delta(t'-t) + \frac{a^2 e^{a (t'-t)}-b^2 e^{b (t'-t)}}{a-b}\Theta(t'-t).
\end{align*}
while the order 0 approximation as computed by the original Neumann series is simply $\delta(t'-t)$.
Note that the result above, as well as $T$ and the series representation of $f$ are all actually well defined in the limit $a\to b $ owing to algebraic simplifications; while in the limits $a/b\to 0$ or $b/a\to 0$ the 0th order as given by Corollary~\ref{ReSummedNeumann} becomes exact, as expected.
At order 1, we get 
\begin{align*}
&f^{(1)}(t',t) := \left(1_\ast + T\right)\ast\Big(1_\ast + \frac{a^2 e^{a (t'-t)}-b^2 e^{b (t'-t)}}{a-b}\Theta(t'-t)\Big),\\
&=\delta(t'-t)+\frac{a+b}{a-b} \left(a e^{a (t'-t)}-b e^{b (t'-t)}\right)\Theta(t'-t)\\
&\hspace{5mm}+\frac{a b}{(a-b)^3} e^{ -(a+b)t}  \left(e^{a t+b t'} \left(a^2+a b^2 (t'-t)-b^3
   (t'-t)+b^2\right)\right.\\
   &\hspace{26mm}\left.-e^{a t'+b t} \left(-a^2(a-b) (t'-t)+a^2+b^2\right)\right)\Theta(t'-t).
\end{align*}
Once again, this is well defined in the limit $a\to b$ in spite of the $a-b$ in denominators, owing to algebraic simplifications with the numerators. It is straightforward to continue to higher orders but the resulting expressions are too cumbersome to be reported here. In contrast with the expressions obtained from the original Neumann series, e.g. $\delta(t'-t)+(a+b)$ at order 1, the superiority of Corollary~\ref{ReSummedNeumann} is obvious.  
\end{example}

\begin{example}[Heun functions]
Heun confluent functions are special transcendental functions known from general relativity and astrophysics \cite{Hortacsu2018} as well as quantum optics \cite{xie2010}. Heun's functions are known only from the context of differential equations \cite{Ronveaux1995, DLMF}, being defined as the solution to Heun’s differential equation. No integral equation satisfied by these functions is known so far. To quote a recent review \cite{Hortacsu2018} on Heun's functions, the current state of knowledge is as follows :

\textit{``No example has been given of a solution of Heun’s equation expressed in the form of a definite integral or contour integral involving only functions which are, in some sense, simpler.[...] This statement does not exclude the possibility of having an infinite series of integrals with `simpler' integrands''.}\\[-.5em]

Here we give two Volterra equations of the second kind for which Heun's confluent function is the solution, implying two infinite series of integrals converging to Heun's function via Corollary~\ref{ReSummedNeumann}.\\[-.5em]

We start with the work of Xie and Hai  \cite{xie2010}, who considered the system of coupled linear differential equations with non-constant coefficients, 
\begin{subequations}\label{HeunSystem}
\begin{align}
2i\omega\frac{d}{dt}a(t) &= \nu b(t)+f(t) a(t),\\
2i\omega\frac{d}{dt}b(t) &= \nu a(t)-f(t)b(t),
\end{align}
\end{subequations}
where $f(t):=f_1 \sin(\omega t)$ and $f_1$, $\nu$ and $\omega$ are real parameters originating from a quantum Hamiltonian. The authors showed that this system is equivalent to Heun's equation, hence is solved by certain Heun confluent functions, here denoted $H_c$. For example, for $a(t)$ we have \cite{xie2010},
\begin{align*}
a(t)&=c_1e^{i (f_1/\omega)\sin^2(\omega t/2)}\,H_c\!\left(\alpha,\beta,\gamma,\delta,\eta,\sin^2(\omega t/2)\right)\\
&\hspace{5mm}+c_2e^{i (f_1/\omega)\sin^2(\omega t/2)}\sin(\omega t/2)\,H_c\!\left(\alpha,-\beta,\gamma,\delta,\eta,\sin^2(\omega t/2)\right),
\end{align*}
where $c_1$ and $c_2$ are constants determined from the initial conditions and $\alpha:=2i f_1/\omega$, $\beta=\gamma=-1/2$, $\delta=if_1/\omega$ and $\eta=-if_1/(2\omega)+3/8-\nu^2/(4\omega^2)$.\\[-1em]

The system of Eq.~(\ref{HeunSystem}) is easily mapped into a homogenous linear Volterra equation of the second kind via the method of path-sum \cite{Giscard2015}. This gives, for $c_1=1$ and $c_2=0$, 
\begin{align}\label{equationHeun}
a(t)|_{c_1=1,c_2=0}&= \int_0^t \Big(1_\ast +(i/2)f -(\nu^2/4)R \Big)^{\ast-1}\!(\tau,0) \,d\tau,
 \end{align}
 where $R:=1\ast F\ast1=\int_t^{t'}\int_{\tau_1}^{t'}F(\tau_2,\tau_1)d\tau_2d\tau_1\,\Theta$ and 
$$ 
F(t',t) := \big(1_\ast - (i/2) f\big)^{\ast-1} =\delta(t'-t)+\frac{i}{2} \sin (\omega t') e^{-i/(2 \omega)
   \big(\cos(\omega t')-\cos(\omega t)\big)}\Theta,
$$    
as per Proposition~\ref{SingleSeparable}. Here again and from now on, we omit the $(t'-t)$ arguments of the Heaviside theta functions to alleviate the equations.
The part of the solution with $c_1=0$ and $c_2=1$ is obtained from the above by $\ast$-multiplication of $\dot{a}(t)_{c_1=1,c_2=0}$ with $q(t',t):=(\nu/2)\int_{t}^{t'}F(t',\tau)d\tau\,\Theta$, as dictated again by path-sum \cite{Giscard2015}.

The integral equation of interest here is that satisfied by $\dot{a}$, namely
\begin{equation}\label{VolterraHeun}
\dot{a} = 1_\ast +\left(-(i/2)f +(\nu^2/4)R\right) \ast \dot{a}
\end{equation}
This is precisely a homogeneous linear Volterra integral equation of the second kind with separable kernel $K(t',t) := \tilde{K}(t',t)\Theta(t'-t)$, where  
$$
\tilde{K}(t',t)=-\frac{i}{2}f(t')+\frac{\nu^2}{4}s(t')-\frac{\nu^2}{4}s(t),
$$
with $s:=\int \int_{\tau_1}^{t'}F(\tau_2,\tau_1)d\tau_2d\tau_1$ so that $R(t',t)=\big(s(t')-s(t)\big)\Theta$. Already this is in important result in the study of Heun functions for which not only was there no known Volterra equation whose solution involved Heun functions, but also because now that such an equation is known the ordinary Neumann series $\sum_{k=0}^\infty K^{\ast k}$ provides a novel series representation of these functions. 
For us, it is convenient to see the kernel $K$ as the sum of two kernels $K_i(t',t)=\tilde{K}_i(t',t)\Theta(t'-t)$, $i=1,2$, with 
\begin{align*}
\tilde{K}_1(t',t)&\equiv \tilde{K}_1(t'):=-\frac{i}{2}f(t')+\frac{\nu^2}{4}s(t'),\\
\tilde{K}_2(t',t)&\equiv\tilde{K}_2(t):=-\frac{\nu^2}{4}s(t).
\end{align*}
Indeed, given that $\tilde{K}_1(t')$ and $\tilde{K}_2(t)$ effectively depend on a single time variable each, the corresponding resolvents $R_{K_1}$ and $R_{K_2}$ are given by Proposition.~(\ref{SingleSeparable}):
\begin{align*}
R_{K_1}(t',t)&=\delta(t'-t)+\left(-\frac{i}{2}f(t')+\frac{\nu^2}{4}s(t')\right)e^{-\frac{i}{2}\int_{t}^{t'}f(\tau)d\tau+\frac{\nu^2}{4}\int_{t}^{t'} s(\tau)d\tau}\Theta,\\
R_{K_2}(t',t)&=\delta(t'-t)-\frac{\nu^2}{4}s(t)e^{-\frac{\nu^2}{4}\int_{t}^{t'} s(\tau) d\tau}\Theta,
\end{align*}
where $\Theta\equiv \Theta(t'-t)$. For convenience, in the following we designate $\tilde{R}_{K_i}(t',t)$ the ordinary part of the $R_{K_i}$ so that $R_{K_i}=\delta(t'-t)+\tilde{R}_{K_i}(t',t)\Theta(t'-t)$.

We can now turn to using Corollary~\ref{ReSummedNeumann} with $d=2$, gaining yet another new series representation of Heun functions, so far expressed solely through a power series expansion around the origin \cite{Slavyanov2000, Ronveaux1995}. 

In order to alleviate the notation, we introduce $\tilde{R}_{12}(t',t)$ such that $R_{K_1}\ast R_{K_2}= \delta(t'-t) + \tilde{R}_{12}(t',t)\Theta(t'-t)$ and thus
\begin{align*}
\tilde{R}_{12}(t',t)&=\tilde{R}_{K_1}(t',t)+\tilde{R}_{K_2}(t',t)+\int_{t}^{t'}\tilde{R}_{K_1}(t',\tau)\tilde{R}_{K_2}(\tau,t)d\tau,\\
&=\left(-\frac{i}{2}f(t')+\frac{\nu^2}{4}s(t')\right)e^{-\frac{i}{2}\int_{t}^{t'}f(\tau)d\tau+\frac{\nu^2}{4}\int_{t}^{t'} s(\tau)d\tau}\\
&\hspace{10mm}-\frac{\nu^2}{4}s(t)e^{-\frac{\nu^2}{4}\int_{t}^{t'} s(\tau) d\tau}\\
&\hspace{15mm}+\left(-\frac{i\nu^2}{8} f(t')s(t)-\frac{\nu^4}{16}s(t')s(t)\right)\times\\
&\hspace{20mm}\int_{t}^{t'}e^{-\frac{i}{2}\int_{\tau'}^{t'}f(\tau)d\tau+\frac{\nu^2}{4}\int_{\tau'}^{t'} s(\tau)d\tau} e^{-\frac{\nu^2}{4}\int_{t}^{\tau'} s(\tau) d\tau}d\tau'.
\end{align*}
Then the quantity $T=1_\ast-R_{K_2}\ast R_{K_1}\ast (1_\ast - K_1-K_2)$ is here found to be $T(t',t)=\tilde{T}(t',t)\Theta(t'-t)$ with 
\begin{align*}
\tilde{T}(t',t)&=\tilde{R}_{12}(t',t)+\int_{t}^{t'}\tilde{R}_{12}(t',\tau)\left(\tilde{K}_1(\tau,t)+\tilde{K}_2(\tau,t)\right)d\tau.
\end{align*}
With this, Corollary~\ref{ReSummedNeumann}, Remark~\ref{Order} and Eq.(~\ref{equationHeun}) give 
\begin{align*}
&e^{i (f_1/\omega)\sin^2(\omega t/2)}\,H_c\!\left(\alpha,\beta,\gamma,\delta,\eta,\sin^2(\omega t/2)\right)\Theta(t)\\
&\hspace{10mm}=\int_{0}^t\left(\sum_{k=0}^{\infty} T^{\ast k}\ast K_1\ast K_2\right)(\tau',0)d\tau',\\
&\hspace{10mm}=\int_{0}^t\Big(\left(1_\ast + T + T^{\ast 2}+T^{\ast 3} + \cdots\right)\ast K_1\ast K_2\Big)(\tau',0)d\tau'.
\end{align*}
Because $\tilde{K_1}$ depends only on $t'$ while $\tilde{K}_2$ depends only on $t$, $K_1\ast K_2$ simplifies to  $(K_1\ast K_2)(t',t)= \tilde{K}_1(t')\tilde{K}_2(t)(t'-t)\Theta(t'-t)$ and therefore
\begin{align*}
&e^{i (f_1/\omega)\sin^2(\omega t/2)}\,H_c\!\left(\alpha,\beta,\gamma,\delta,\eta,\sin^2(\omega t/2)\right)\\
&\hspace{3mm}=\int_{0}^{t}\tilde{K}_1(\tau')\tau' d\tau'\,\tilde{K}_2(0)\\
&\hspace{5mm}+\int_{0}^{t}\int_{0}^{\tau'}\tilde{T}(\tau',\tau_1) \tilde{K}_1(\tau_1)\tau_1\, d\tau_1 d\tau'\, \tilde{K}_2(0)\\
&\hspace{5mm}+\int_{0}^t \int_{0}^{\tau'}\int_{\tau_1}^{\tau'}\tilde{T}(\tau',\tau_2)\tilde{T}(\tau_2,\tau_1)\tilde{K}_1(\tau_1)\tau_1\, d\tau_2 d\tau_1 d\tau'\,\tilde{K}_2(0)\\
&\hspace{5mm}+\int_{0}^t \int_{0}^{\tau'}\int_{\tau_1}^{\tau'}\int_{\tau_2}^{\tau'}\tilde{T}(\tau',\tau_3)\tilde{T}(\tau_3,\tau_2)\tilde{T}(\tau_2,\tau_1)\tilde{K}_1(\tau_1)\tau_1\,d\tau_3 d\tau_2 d\tau_1 d\tau'\,\tilde{K}_2(0)\\
&\hspace{5mm}+\cdots
\end{align*} 
here we have removed the $\Theta(t)$ factors appearing in all terms on both the left and right hand-sides, however this means that the above expansion is valid only for $t\geq 0$. 
\end{example}

\section{Conclusion}
In this note, we presented a novel, analytically computable series for the solution of inhomogenous linear Volterra equations of the second kind with arbitrary sum kernel. Equivalently, we expressed the solution of the equation $f=g+K\ast f$ with kernel $K=\sum_i K_i$ in terms of the solutions of the equations $f_i=\delta + K_i \ast f_i$. These results 
completely bypass the standard strategy consisting in turning the Volterra equation in a system of coupled linear differential equations with non-constant coefficients, with the aim of producing explicit, computable, analytical expressions. As an illustration, we obtained an hitherto unknown Volterra linear integral equation of the second kind satisfied by Heun's confluent functions and a novel series representation for these functions.

\bibliographystyle{plain}

\end{document}